\documentclass[runningheads,a4paper]{llncs}
\usepackage{algpseudocode}
\usepackage{algorithm}
\usepackage{amssymb}
\usepackage{graphicx}
\usepackage{caption}
\usepackage{subcaption}
\usepackage{tcolorbox}
\usepackage{fancyvrb}
\usepackage{tcolorbox}

\usepackage{amsmath}
\usepackage[plainpages=false]{hyperref}
\usepackage{algpseudocode}
\usepackage{algorithm}
\usepackage{xspace}  
\usepackage{braket}
\usepackage[utf8]{inputenc}
\usepackage{csquotes}
\usepackage[inline]{enumitem}
\usepackage{mathtools}
\usepackage{todonotes}
\usepackage{xargs}
\usepackage{nicefrac}

\usepackage{tikz}
\usetikzlibrary{decorations.pathreplacing,calc}

\usepackage{url}
\urldef{\mailsa}\path|{alfred.hofmann, ursula.barth, ingrid.haas, frank.holzwarth,|
\urldef{\mailsb}\path|anna.kramer, leonie.kunz, christine.reiss, nicole.sator,|
\urldef{\mailsc}\path|erika.siebert-cole, peter.strasser, lncs}@springer.com|    
\newcommand{\keywords}[1]{\par\addvspace\baselineskip
\noindent\keywordname\enspace\ignorespaces#1}

\begin{document}

\mainmatter  

\title{Online Connected Dominating Set Leasing\thanks{This work was partially supported the Federal Ministry of Education and Research (BMBF) as part of the project `Resilience by Spontaneous Volunteers Networks for Coping with Emergencies and Disaster' (RESIBES), (grant no 13N13955 to 13N13957).}}

\titlerunning{Online Connected Dominating Set Leasing}

%
%
\author{Christine Markarian}
\authorrunning{Christine Markarian}

\urldef{\mailsa}\path|chrissm@mail.uni-paderborn.de|    
\institute{
	Heinz Nixdorf Institute \& Computer Science Department,\\
	Paderborn University, F\"urstenallee 11, 33102 Paderborn, Germany\\
	\mailsa}
\toctitle{Lecture Notes in Computer Science}
\tocauthor{Authors' Instructions}
\maketitle

\begin{abstract}

We introduce the \emph{Online Connected Dominating Set Leasing} problem (OCDSL) in which we are given an undirected connected graph $G = (V, E)$, a set $\mathcal{L}$ of lease types each characterized by a duration and cost, and a sequence of subsets of $V$ arriving over time. A node can be leased using lease type $l$ for cost $c_l$ and remains active for time $d_l$. The adversary gives in each step $t$ a subset of nodes that need to be dominated by a connected subgraph consisting of nodes active at time $t$. The goal is to minimize the total leasing costs. OCDSL contains the \emph{Parking Permit Problem}~\cite{PPP} as a special subcase and generalizes the classical offline \emph{Connected Dominating Set} problem~\cite{Guha1998}. It has an $\Omega(\log ^2 n + \log |\mathcal{L}|)$ randomized lower bound resulting from lower bounds for the \emph{Parking Permit Problem} and the \emph{Online Set Cover} problem~\cite{Alon:2003:OSC:780542.780558,Korman}, where $|\mathcal{L}|$ is the number of available lease types and $n$ is the number of nodes in the input graph. We give a randomized $\mathcal{O}(\log ^2 n + \log |\mathcal{L}| \log n)$-competitive algorithm for OCDSL. We also give a deterministic algorithm for a variant of OCDSL in which the dominating subgraph need not be connected, the \emph{Online Dominating Set Leasing} problem. The latter is based on a simple primal-dual approach and has an $\mathcal{O}(|\mathcal{L}| \cdot \Delta)$-competitive ratio, where $\Delta$ is the maximum degree of the input graph.

\end{abstract}

\section{Introduction}

Dominating Set problems, where the goal is to find a minimum subgraph of a given (undirected) graph such that each node is either in the subgraph or has an adjacent node in it, form a fundamental class of optimization problems that have received significant attention in the last decades. The \emph{Connected Dominating Set} problem (CDS) - which asks for a minimum such subgraph that is connected (a connected dominating set) - is one of the most well-studied problems in this class with a wide range of applications in wireless networks~\cite{Du2013}. It is known to be $\mathcal{NP}$-complete even in planar graphs~\cite{GareyJ79} and admits an $\mathcal{O}(\ln \Delta)$-approximation for general graphs, where $\Delta$ is the maximum node degree of the input graph~\cite{Guha1998}. The latter is the best possible unless $\mathcal{NP} \subset DTIME(n^{\log n \log n})$. Connected dominating sets have been widely employed in wireless ad-hoc networks by serving as virtual backbones for routing, broadcasting, and connectivity management~\cite{doi:10.1177/155014771703201}. The following two assumptions are often made. (1) Once a node is assigned as a member of the virtual backbone (a dominating node), it can serve forever without incurring further costs in the future. It is natural though to consider costs  - such as maintenance costs - that appear over time. (2) In a classical connected dominating set problem, all nodes are required to be served (dominated) at once. Nevertheless, this need not be true in many scenarios (e.g., arrival of new clients in a client-server scenario). In this paper, we address these two concerns by providing \emph{leasing} options to dominating nodes such that a node can be leased for different durations and costs - once the lease for a node expires, it needs to be paid for again. Moreover, to capture the dynamic nature of wireless ad-hoc networks, we consider CDS in an \emph{online} setting. This is formulated as the \emph{Online Connected Dominating Set Leasing} problem (OCDSL), defined as follows. 

\begin{definition} (OCDSL) Given an undirected connected graph $G = (V, E)$, a set $\mathcal{L}$ of lease types each characterized by a duration and cost, and a sequence of subsets of $V$ arriving over time. A node leased using lease type $l$ incurs cost $c_l$ and remains active for time $d_l$. A node is said to be dominated by a subset $S$ of $V$ if it is either in $S$ or has an adjacent node in $S$. In each step $t$, the adversary gives a subset of nodes of $V$ that need to be dominated by a connected subgraph consisting of nodes active at time $t$. The goal is to minimize the total leasing costs. 
\end{definition}

OCDSL contains the \emph{Parking Permit} problem~\cite{PPP} (PP) as a special subcase. In his seminal work~\cite{PPP}, Meyerson has introduced the first theoretical leasing model with this problem, defined as follows. Each day, depending on the weather, we have to either use the car (if it is rainy) or walk (if it is sunny). In the former case, we must have a valid parking permit, which we choose among a set $\mathcal{L}$ of different types of permits (leases), each having a different duration and price. At any time, lease prices respect \emph{economy of scale} such that a longer lease costs less per unit time. The goal is to buy a set of leases in order to cover all rainy days while minimizing the total cost of purchases and without using weather forecasts. There have been a series of classical complex optimization problems~\cite{InfraLeasingProb} studied with the leasing notion, such as the \emph{Set Cover} problem~\cite{JOCO}, the \emph{Facility Location} problem~\cite{OffOnFacilityLeasing}, and the \emph{Steiner Forest} problem~\cite{WADSMarcin}.

There have been a vast number of works on CDS and its variants. For general undirected graphs, Guha \emph{et al.}~\cite{Guha1998} have proposed a greedy approach that yields an $\mathcal{O}(\ln \Delta)$-approximation for CDS, where $\Delta$ is the maximum degree of the input graph. For general directed graphs, Li \emph{et al.}~\cite{Li:2009:CSC:1502808.1502872} have given an $\mathcal{O}(\ln n)$-approximation for the \emph{Strongly Connected Dominating Absorbent Set} problem (SCDAS) - which asks for a minimum subset $S$ of nodes such that each node has both an in- and out- neighbor in $S$ and the subgraph induced by $S$ is strongly connected, where $n$ is the number of nodes in the input graph. Connected dominating sets have been intensively studied as subproblems in wireless ad-hoc networks, which are commonly modeled as unit disk graphs. Lichtenstein~\cite{Lichte} has
shown that CDS is $\mathcal{NP}$-complete for unit disk graphs. Constant-factor approximation algorithms for the latter were given by Marathe \emph{et al.}~\cite{Marathe}. 
\label{sec:relatedwork}
\label{sec:relatedwork}

\section{Preliminaries}
\label{sec:preliminaries}

We assume the following lease configuration for the nodes.  

\begin{definition} (Lease Configuration)
\label{leaseconfig}
A node can be leased with lease type $l$ only at times multiple of $d_l$, that is, at times $t$ with $t \equiv 0 \mod d_l$. Moreover, all lease lengths are power of two.
\end{definition}
This implies that at any time $t$ and any node $u$ there are exactly $|\mathcal{L}|$ lease types corresponding to $u$. This configuration has been similarly defined by Meyerson for the \emph{Parking Permit Problem}~\cite{PPP}. Meyerson has shown that by assuming this configuration (called the \emph{Interval Model}) one loses a constant factor in the competitive ratio. Similar arguments were given for all generalizations of the \emph{Parking Permit Problem}~\cite{JOCO,WADSMarcin,OffOnFacilityLeasing}. 

\paragraph {\bf Lower bound.} There is an $\Omega(\log ^2 n + \log |\mathcal{L}|)$ randomized lower bound for OCDSL resulting from two lower bounds, the first of which comes from a lower bound for the \emph{Parking Permit Problem}~\cite{PPP}, a special subcase of OCDSL. Meyerson has shown a randomized $\Omega(\log (|\mathcal{L}|)$ lower bound for the \emph{Parking Permit Problem}. The second comes from a lower bound for the \emph{Online Set Cover} problem introduced by Alon \emph{et al.}~\cite{Alon:2003:OSC:780542.780558}. Korman~\cite{Korman} has shown a randomized $\Omega(\log n \log m)$ lower bound for the \emph{Online Set Cover} problem - where $n$ is the number of elements and $m$ is the number of subsets. The offline \emph{Set Cover} problem reduces to the offline \emph{Connected Dominating Set} problem~\cite{Guha1998}. It is easy to see that, by similar argumentation, a reduction between the corresponding online counterparts can be made, where the online variant of CDS can be seen as a special case of OCDSL in which there is one lease type of infinite duration. Hence, the lower bound for OCDSL follows. 

\begin{figure}
\centering
\begin{subfigure}{.45\textwidth}
  \centering
  \includegraphics[width=1\linewidth]{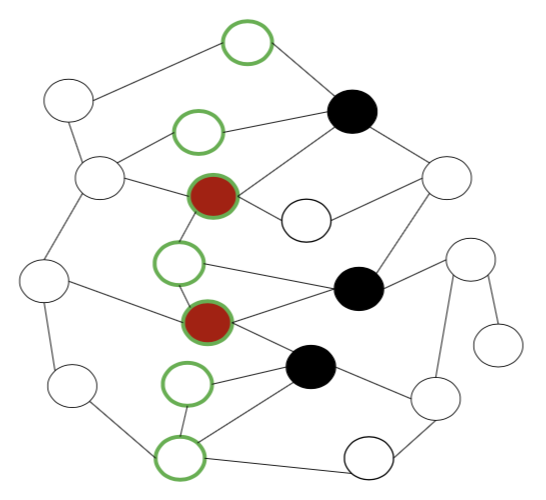}
  \caption{Phase 1}
  \label{fig:phase1}
\end{subfigure}%
\begin{subfigure}{.45\textwidth}
  \centering
  \includegraphics[width=1\linewidth]{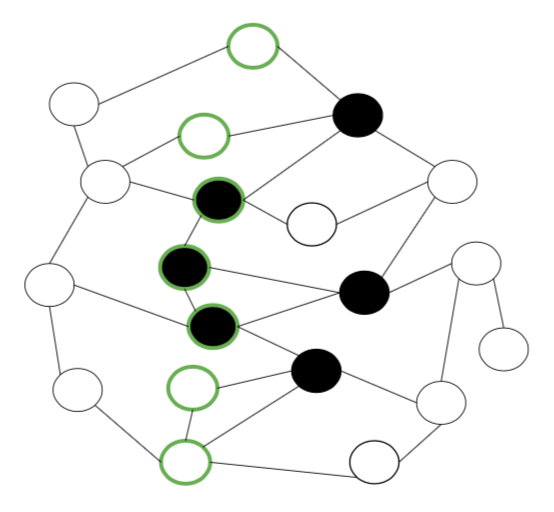}
  \caption{Phase 2}
  \label{fig:phase2}
\end{subfigure}
\caption{The two phases of the algorithm on a time step $t$. Phase 1: Nodes with green border form the set $D_t$, black nodes form the set $S_t$, and red nodes form the set $R_t$. Phase 2: Black nodes form a connected dominating set for the set $D_t$.}
\label{fig:phases}
\end{figure}

\section{Online Algorithm}
\label{sec:onlinealgorithm}

In this section we present a randomized algorithm for OCDSL. 

The algorithm is given an undirected connected graph $G = (V, E)$, a set $\mathcal{L}$ of lease types each characterized by a duration and cost, and a sequence of subsets of nodes arriving over time. A node $u$ leased with type $l$ at time $t$ is denoted as a triplet $(u, l, t)$. The set of all triplets is denoted as $\mathcal{V}$. A triplet whose lease intersects time $t$ is called a \emph{$t$-triplet}. A triplet is \emph{active} when its cost has been paid for. The structure of the algorithm is similar to the classical structure used in most offline algorithms for CDS - first construct a dominating set and then connect it with a Steiner tree. Suppose the adversary gives a subset $D_t$ of nodes at time $t$. We denote by $\mathcal{D}$ the collection of all these subsets. The algorithm considers only $t$-triplets and runs in two phases. In the first phase it purchases leases of a subset $S_t$ of nodes that dominate $D_t$, that is, each node in $D_t$ is either in $S_t$ or has a neighbor in it. In the second phase it makes sure that the subset induced by $S_t$ is connected by active triplets, by purchasing additional $t$-triplets when necessary (See Fig \ref{fig:phases} for an illustration). The first phase of the algorithm is based on \emph{randomized rounding}, an approach typically used in online algorithms. A fractional solution in which a fraction is assigned to each triplet is first constructed. The latter is then rounded to yield an integral solution. The second phase is based on running Meyerson's $\mathcal{O}(\log n \log (|\mathcal{L}|))$-competitive algorithm for the \emph{edge-weighted Online Steiner Forest Leasing} problem. The algorithm assigns to each triplet $(v, l, t)$ a weight $w_{(v, l, t)}$ initially set to zero and non-decreasing throughout the algorithm. A random variable $\mu$ is chosen as the minimum among  $2\left\lceil \log(n+1) \right\rceil$ independent random variables, distributed uniformly in the interval $\left[0, 1\right]$. We call the set of triplets that can dominate an arriving node $u$ the \emph{dominators} of $u$ and denote it as $W_u$. Given a subset $D_t$ of nodes, the algorithm runs the two phases depicted in Algorithm \ref{alg:1}.

\begin{algorithm}
  \caption{}
  \label{alg:1}
  \begin{algorithmic}
  \State {\bf Phase 1} (Fig \ref{fig:phase1}) \
  \State i) For each $u \in D_t$, if $u$ does not have an active dominator,
     \State \hspace{6pt} a) While the total weight of $u$'s dominators is less than 1,
     \State \hspace{17pt} Set the weight $w_{(i,l,t')}$ of each dominator to:  $$w_{(i,l,t')} \cdot (1 + 1/c_{l}) + \frac{1}{|W_{u}| \cdot |\mathcal{L}| \cdot c_{l}}$$\  
    \State \hspace{6pt} b) Purchase $(i,l,t')$ if $w_{(i,1,t')} > \mu$.\
    \State \hspace{6pt} c) If $u$ does not have an active dominator, choose a dominator of smallest lease type and purchase it.\
    \State ii) Assign to each $u \in D_t$ an active dominator. Let $S_t$ denote the set of these dominators. For each triplet $(i,l,t') \in S_t$, purchase a triplet of cheapest lease type corresponding to one of the nodes in $D_t$ that is dominated by $(i,l,t')$ - call it a \emph{representative}. To choose these nodes, use a greedy approach - choose a representative $u \in D_t$ with more adjacent dominators in $S_t$ first. \\
    \hrulefill
   \State {\bf Phase 2} (Fig \ref{fig:phase2}) \
   \State i) Choose arbitrarily one of the representatives from Step 1 - call it the \emph{root}. Let $R_t$ be the set of representatives that are not connected to the root via a path consisting of active $t$-triplets. \
   \State ii) Assign a unit-weight to all edges and run the algorithm for edge-weighted OSFL with the set $R_t$ as terminals for time $t$. The algorithm outputs a set of leased edges. \
    \State iii) For each edge $(v, w)$ leased (at the current step $t$) with type $l$ and starting time $t'$, purchase the corresponding triplets for the nodes $v$ and $w$ with lease type $l$ and starting time $t'$.
    \end{algorithmic}
\end{algorithm}

\paragraph{Correctness.} For a given subset $D_t$ of nodes, the algorithm guarantees in Step c that each node in $D_t$ has at least one dominator active at time $t$. Each is then assigned to arbitrary one active dominator. To connect the set of these dominators at time $t$, the algorithm buys a representative node for each dominator from $D_t$ (adjacent to it) and then connects the representative nodes with additional $t$-triplets forming a Steiner tree. Hence the dominators induce a connected subgraph that consist of active $t$-triplets.

\section{Analysis}
\label{sec:analysis}
In this section we show that Algorithm \ref{alg:1} has an $\mathcal{O}(\log ^2 n + \log |\mathcal{L}| \log n)$-competitive ratio, where $n$ is the number of nodes in the input graph and $|\mathcal{L}|$ is the number of lease types. 

Let $Opt$ be the cost of an optimal offline solution and let $\Delta$ be the maximum degree of the input graph. Clearly, we have that $\Delta \leq n$. Since the algorithm is randomized, we measure its cost in expectation. We analyze each phase separately. We denote by $C_1$ and $C_2$ the algorithm's cost in Phase 1 and Phase 2, respectively.

\paragraph{\bf Phase 1.} For a subset $D_t$ of nodes arriving at time $t$ the algorithm constructs a dominating set consisting of nodes active at time $t$. To this end, it constructs a fractional solution (in step $a$), which is rounded to an integral solution (in step $b$). To guarantee that all arriving nodes are dominated, the algorithm purchases further leases if necessary (in step $c$). Moreover, the algorithm buys additional leases called representatives (in step $ii$), which are used in Phase 2 to connect active dominators of Step 1. The following lemma bounds the cost of the fractional solution in terms of $Opt$.

\begin{lemma}
\label{lem1}
The cost of the fractional solution constructed by Algorithm \ref{alg:1} in Phase 1 is at most $\mathcal{O}(\log (|\mathcal{L}| \cdot (\Delta +1)))$ times $Opt$.
\end{lemma}
\begin{proof}
We need to bound $\sum_{(i,l,t') \in \mathcal{V}} c_l \cdot w_{(i,l,t')}$. Clearly, the algorithm increases the weights only when there is a node that is not dominated. Let $u$ be such a node. The algorithm increases the weight corresponding to each of $u$'s $\left| W_u \right|$ dominators such that the fractional cost added by each dominator $(i, l, t')$ is $\left(c_l \cdot (\frac{w_{(i,l,t')}}{c_l} + \frac{1}{\left| W_u  \right| \cdot c_l} )\right)$. Thus, the fractional cost added by all dominators is: 
\begin{equation}
\label{fracinc1}
\sum_{(i,l,t') \in W_u} c_l \cdot \left( \frac{w_{(i,l,t')}}{c_l} + \frac{1}{\left| W_u  \right| \cdot c_l} \right) \leq 2
\end{equation}

The above inequality holds since $\sum_{(i,l,t') \in W_u} w_{(i,l,t')} \leq 1$ before any weight increase. An optimal solution must contain at least one triplet $(j,l',t') \in W_u$. The weight of this triplet becomes at least 1 after at most $(c_{l'} \cdot \log{\left| W_u \right|})$ weight increases and hence no further weight increases can be made since the sum of weights of the dominators must not exceed 1. This means that the maximum number of weight increases is $(c_{l'} \cdot \log{\left| W_u \right|})$ and by inequality \ref{fracinc1}, we have that each weight increase costs at most 2. Thus, we conclude that:  

\begin{equation}
\label{fracinc2}
\sum_{(i,l,t') \in W_u} c_l \cdot w_{(i,l,t')} \leq 2 \cdot c_{l'} \cdot \log |W_u|
\end{equation}

This argument holds for all triplets in the optimal solution. Adding up over all these triplets leads to the following.

\begin{equation}
\label{fffff}
\sum_{(i,l,t') \in \mathcal{V}} c_l \cdot w_{(i,l,t')} \leq 2 \cdot Opt \cdot \log |W_u|
\end{equation}

\end{proof}

Next, we bound the cost of the integral solution in terms of the cost of the fractional solution.

\begin{lemma}
\label{lem2}
The cost of the integral solution constructed by Algorithm \ref{alg:1} in Phase 1 is at most $\mathcal{O}(\log n)$ times the cost of the fractional solution constructed by Algorithm \ref{alg:1} in Phase 1.   
\end{lemma}
\begin{proof}

Notice that the algorithm purchases triplets either in Step b, Step c, or Step ii of Phase 1. Let us first observe the expected cost of the integral solution constructed in Step b. Fix a triplet $(i, l, t) \in \mathcal{V}$ and $q: 1 \leq q \leq 2\left\lceil \log(n+1) \right\rceil$. Let $X_{(i, l, t), q}$ be the indicator variable of the event that  $(i, l, t)$ is bought by the algorithm. The expected cost of the integral solution is at most: 
\begin{equation}
\label{frac}
\sum_{(i, l, t) \in \mathcal{V}} \sum_{q = 1}^{2\left\lceil \log(n+1) \right\rceil} c_l \cdot Exp \left[X_{(i, l, t), q} \right]  \leq 2\left\lceil \log(n+1) \right\rceil \cdot \sum_{(i, l, t) \in \mathcal{V}} c_l \cdot w_{(i, l, t)} 
\end{equation}

Now we show that the expected cost of the integral solution constructed in Step c is negligible. Notice that the algorithm purchases a triplet of cheapest/smallest lease only if there is a node that is not dominated. Fix any such node $v$. For a single $1 \leq q \leq 2\left\lceil \log(n+1) \right\rceil$, the probability that $v$ is not dominated is at most:

$$ \prod_{(i, l, t) \in W_v}  (1 - w_{(i, l, t)}) \leq e^{-\sum_{(i, l, t)  \in W_v} w_{(i, l, t)}} \leq 1/e $$

The last inequality holds because $\sum_{(i, l, t)  \in W_v} w_{(i, l, t)}\geq 1$ is guaranteed by the algorithm. Hence, the probability that $v$ is not dominated for all $1 \leq q \leq 2\left\lceil \log(n+1) \right\rceil$ is at most $1/n^2$. Moreover, we have that the cost of the cheapest lease is a lower bound on the cost of the optimal offline solution $Opt$. Therefore, the expected cost for all $n$ nodes is at most $n \cdot 1/n^2 \cdot Opt$. It remains to measure the expected cost of the integral solution constructed in Step ii in which the algorithm purchases the representative set. We show that the expected cost of all representative triplets purchased is at most $\mathcal{O}(\log n)$ times the cost of the fractional solution. Note that all representative triplets are of cheapest lease type ($l = 1$). Fix a triplet $(u, 1, t) \in \mathcal{V}$ and $q: 1 \leq q \leq 2\left\lceil \log(n+1) \right\rceil$. Let $Y_{(u, 1, t), q}$ be the indicator variable of the event that $(u, 1, t)$ is chosen as a representative. This event occurs when the algorithm buys a triplet corresponding to one of its dominators. Let $W_u$ be the set of these triplets. The expected cost of representative $(u, 1, t)$ is then: 

\begin{equation}
\label{frac}
c_1 \cdot Exp \left[ Y_{(u, l, t), q} \right] \leq 2\left\lceil \log(n+1) \right\rceil \cdot c_1 \cdot \sum_{(i, l, t') \in W_u} w_{(i, l, t')} 
\end{equation}
Since $c_1$ is less that $c_l$ for all $l \in L$, we have that:   
\begin{equation}
\label{frac}
2\left\lceil \log(n+1) \right\rceil \cdot c_1 \cdot \sum_{(i, l, t') \in W_u} w_{(i, l, t')} \leq 2\left\lceil \log(n+1) \right\rceil \cdot \sum_{(i, l, t') \in W_u} c_l \cdot w_{(i, l, t')}
\end{equation}

The optimal solution must contain at least one triplet in $W_u$ - since it dominates the representative set. Thus, the same argument that led to Equation \ref{fracinc2} holds and we can conclude an overall expected cost of at most $2\left\lceil \log(n+1) \right\rceil$ times the cost of the fractional solution.   \qed
\end{proof}
Lemma \ref{lem1} and Lemma \ref{lem2} imply the cost of Algorithm \ref{alg:1} in Phase 1, as follows.

\begin{equation}
\label{phase1equation}
C_1 \leq \mathcal{O}(\log (|\mathcal{L}| \cdot n) \log n) \cdot Opt 
\end{equation} 

\paragraph{\bf Phase 2.} The algorithm for edge-weighted OSFL by Meyerson outputs a set of leased edges - forming a so-called \emph{Steiner forest}. Let $C_{St}$ be the cost of the leased edges. Since the algorithm purchases for each edge leased two corresponding triplets of the same cost, it pays exactly twice as much as $C_{St}$. Recall that Meyerson's randomized algorithm for edge-weighted OSFL has an $\mathcal{O}(\log n \log (|\mathcal{L}|))$-competitive ratio. Let $Opt_{St}$ be the cost of an optimal Steiner forest. We have that $C_{St} \leq \mathcal{O}(\log n \log (|\mathcal{L}|)) \cdot Opt_{St}$. Moreover, the optimal offline solution connects in each step $t$ the corresponding set of representatives at time $t$, since it dominates this set through a connected subgraph consisting of active $t$-triplets. This means that it forms a Steiner forest and hence $Opt_{St} \leq Opt$. Therefore the cost of Algorithm \ref{alg:1} in Phase 2 is at most:

\begin{equation}
\label{phase2equation}
C_2 \leq \mathcal{O}(\log n \log (|\mathcal{L}|)) \cdot Opt
\end{equation} 
\newline
Equations \ref{phase1equation} and \ref{phase2equation} ultimately lead to the following theorem. 

\begin{theorem}
There is a randomized $\mathcal{O}(\log ^2 n + \log |\mathcal{L}| \log n)$-competitive algorithm for OCDSL, where $n$ is the number of nodes in the input graph and $|\mathcal{L}|$ is the number of available lease types. 
\end{theorem}

\section{Online Dominating Set Leasing}
\label{sec:ODSL}

If we do not require the subgraph induced by the dominators to be connected, then Step i of Phase 1 outputs a feasible solution for this special case - the \emph{Online Dominating Set Leasing} problem (ODSL). This results in a randomized $\mathcal{O}(\log n \log (|\mathcal{L}| \cdot \Delta))$-competitive algorithm for ODSL and coincides with the result for the leasing variant of the \emph{Set Cover} problem~\cite{JOCO}. Moreover, one can get a deterministic $\mathcal{O}(|\mathcal{L}| \cdot \Delta)$-competitive algorithm for ODSL, using a simple primal-dual approach.

\begin{theorem}
There is a deterministic primal-dual algorithm for ODSL with $\mathcal{O}(|\mathcal{L}| \cdot \Delta)$-competitive ratio, where $\Delta$ is the maximum degree of the input graph and $|\mathcal{L}|$ is the number of available lease types. 
\end{theorem}

We formulate the problem as a primal-dual program (see Figure \ref{fig:primaldual}). Each triplet $(i,l,t)$ is associated with a primal variable $X_{(i,l,t)}$ that indicates whether it is bought ($= 1$) or not ($= 0$). The main idea of the algorithm is as follows. Whenever a node $u \in \mathcal{D}$ appears, we increase the dual variables corresponding to its denominators until the primal constraint corresponding to $u$ is satisfied, without violating any of the dual constraints. The proof of the competitive ratio is based on the Weak Duality theorem, commonly used in the analysis of classical primal-dual algorithms. No constraints are violated here and so the proof can easily be deduced. Intuitively, the ratio $(|\mathcal{L}| \cdot \Delta)$ comes from the total number of dominators for any given node.

\begin{figure}
\begin{center}
\rule{\textwidth}{0.2pt}
$\underline{\min} \sum\limits_{(i,l,t)\in \mathcal{V}}X_{(i,l,t)}\cdot c_{l}$
\begin{align*}
\text{{\bf Subject to: }}
&\forall u \in \mathcal{D}: \sum\limits_{(i,l,t)\in W_u} X_{(i,l,t)} \ge 1 \\
&\forall (i,l,t)\in \mathcal{V}: X_{(i,l,t)} \in \{0, 1\} 
\end{align*}
\end{center}
\begin{center}
\rule{\textwidth}{0.2pt}
$\underline{\max} \sum\limits_{u \in \mathcal{D}} Y_{u}$ 
\begin{align*}
\text{{\bf Subject to: }}
&\forall (i,l,t)\in \mathcal{V}, \forall u \in \mathcal{D}: \sum\limits_{(i,l,t) \in W_u} Y_{u} \le c_{l}\\
&\forall u \in \mathcal{D}: Y_{u} \ge 0 
\end{align*}
\end{center}
\rule{\textwidth}{0.2pt}
\caption{Primal-dual formulation of ODSL}\label{figLP}
\label{fig:primaldual}
\end{figure}

\section{Open Problems}
\label{sec:conclusion}

We initiate with this work the online leasing study of connected dominating sets. Many variations of connected dominating sets capture interesting scenarios in wireless ad-hoc networks. It would be interesting to extend these problems to their leasing variants. 

Another next step would be to consider other graph classes such as geometric graphs and bounded-degree graphs. Moreover, our proposed algorithm for OCDSL does not extend to directed graphs because of Phase 2 in which we require an algorithm for the \emph{edge-weighted Online Steiner Forest Leasing} problem. The existing algorithm for the latter is based on embedding the input graph into a random tree such that each edge weight is increased by at most $\mathcal{O}(\log n)$, where $n$ is the number of nodes in the graph~\cite{PPP}. Hence, it does not seem to work for directed graphs. Therefore, it would be interesting to come up with another approach that solves the \emph{edge-weighted Online Steiner Forest Leasing} problem for directed graphs. This consequently extends our algorithm for OCDSL to directed graphs.     

Closing the gap between the current upper and lower bounds for OCDSL would be an interesting next step. Note that similar gaps exist for the leasing variants of the \emph{Set Cover} problem~\cite{JOCO} and the \emph{Steiner Forest} problem~\cite{PPP}. Existing algorithms for these three problems are all randomized. One may want to consider deterministic approaches for these problems. 

Nodes in our model are assumed to all have unit weight. In many applications, it might be useful to have different weights on the nodes. Our result for OCDSL is based on running an algorithm for the \emph{edge weighted Steiner Forest Leasing} problem. The latter cannot be transferred to the node-weighted variant since the problem becomes more general for node-weighted graphs (replace each edge of weight $w$ by a node of weight $w$). In fact, there have been many studies on online variants of the \emph{Steiner Forest} problem and its variants for node-weighted graphs. Naor \emph{et al.}~\cite{Naor:2011:ONS:2082752.2082953} have proposed a randomized $\mathcal{O}(\log ^7 k \log^3 n)$-competitive algorithm for the online node-weighted \emph{Steiner Forest} problem, where $k$ is the number of terminals and $n$ is the number of nodes in the input graph. Hajiaghayi \emph{et al.}~\cite{6686192} have later improved the latter by introducing a randomized $\mathcal{O}(\log ^2 k \log n)$-competitive algorithm. It is still not clear whether these results can be extended to the leasing variant of node-weighted \emph{Steiner Forest} problem. If the latter were possible and a poly-logarithmic competitive ratio is achieved, our proposed algorithm would yield a feasible solution for node-weighted graphs. However, a different analysis technique will be needed to argue about the competitiveness of the algorithm - since the argument for the representative set fails in case of weighted nodes.

\bibliographystyle{plain}

\bibliography{Bibliography}

\begin{thebibliography}{10}

\bibitem{JOCO}
Sebastian Abshoff, Peter Kling, Christine Markarian, Friedhelm Meyer auf~der
  Heide, and Peter Pietrzyk.
\newblock Towards the price of leasing online.
\newblock {\em Journal of Combinatorial Optimization}, pages 1--20, 2015.

\bibitem{Alon:2003:OSC:780542.780558}
Noga Alon, Baruch Awerbuch, and Yossi Azar.
\newblock The online set cover problem.
\newblock In {\em Proceedings of the Thirty-fifth Annual ACM Symposium on
  Theory of Computing}, STOC '03, pages 100--105, New York, NY, USA, 2003. ACM.

\bibitem{InfraLeasingProb}
Barbara~M. Anthony and Anupam Gupta.
\newblock Infrastructure leasing problems.
\newblock In Matteo Fischetti and David~P. Williamson, editors, {\em Integer
  Programming and Combinatorial Optimization}, pages 424--438, Berlin,
  Heidelberg, 2007. Springer Berlin Heidelberg.

\bibitem{WADSMarcin}
Marcin Bienkowski, Artur Kraska, and Pawe{\l} Schmidt.
\newblock A deterministic algorithm for online steiner tree leasing.
\newblock In Faith Ellen, Antonina Kolokolova, and J{\"o}rg-R{\"u}diger Sack,
  editors, {\em Algorithms and Data Structures}, pages 169--180, Cham, 2017.
  Springer International Publishing.

\bibitem{Du2013}
Hongjie Du, Ling Ding, Weili Wu, Donghyun Kim, Panos~M. Pardalos, and James
  Willson.
\newblock {\em Connected Dominating Set in Wireless Networks}, pages 783--833.
\newblock Springer New York, New York, NY, 2013.

\bibitem{doi:10.1177/155014771703201}
Deqian Fu, Lihua Han, Zifen Yang, and Seong~Tae Jhang.
\newblock A greedy algorithm on constructing the minimum connected dominating
  set in wireless network.
\newblock {\em International Journal of Distributed Sensor Networks},
  12(7):1703201, 2016.

\bibitem{GareyJ79}
Michael~R. Garey and David~S. Johnson.
\newblock {\em {Computers and Intractability, A Guide to the Theory of
  {NP}-Completeness}}.
\newblock W.H. Freeman and Company, New York, 1979.

\bibitem{Guha1998}
S.~Guha and S.~Khuller.
\newblock Approximation algorithms for connected dominating sets.
\newblock {\em Algorithmica}, 20(4):374--387, Apr 1998.

\bibitem{6686192}
M.~T. Hajiaghayi, V.~Liaghat, and D.~Panigrahi.
\newblock Online node-weighted steiner forest and extensions via disk
  paintings.
\newblock In {\em 2013 IEEE 54th Annual Symposium on Foundations of Computer
  Science}, pages 558--567, Oct 2013.

\bibitem{Korman}
Simon Korman.
\newblock {On the Use of Randomization in the Online Set Cover Problem}.
\newblock Master's thesis, Weizmann Institute of Science, Israel, 2005.

\bibitem{Li:2009:CSC:1502808.1502872}
Deying Li, Hongwei Du, Peng-Jun Wan, Xiaofeng Gao, Zhao Zhang, and Weili Wu.
\newblock Construction of strongly connected dominating sets in asymmetric
  multihop wireless networks.
\newblock {\em Theor. Comput. Sci.}, 410(8-10):661--669, March 2009.

\bibitem{Lichte}
David Lichtenstein.
\newblock Planar formulae and their uses.
\newblock 11:329--343, 05 1982.

\bibitem{Marathe}
Madhav Marathe, Heinz Breu, H~B.~Hunt~Iii, S~Ravi, and Daniel Rosenkrantz.
\newblock Simple heuristics for unit disk graphs.
\newblock 25, 03 1995.

\bibitem{PPP}
Adam Meyerson.
\newblock The parking permit problem.
\newblock In {\em 46th Annual {IEEE} Symposium on Foundations of Computer
  Science {(FOCS} 2005), 23-25 October 2005, Pittsburgh, PA, USA, Proceedings},
  pages 274--284, 2005.

\bibitem{OffOnFacilityLeasing}
Chandrashekhar Nagarajan and David~P. Williamson.
\newblock Offline and online facility leasing.
\newblock {\em Discrete Optimization}, 10(4):361--370, 2013.

\bibitem{Naor:2011:ONS:2082752.2082953}
Joseph~(Seffi) Naor, Debmalya Panigrahi, and Mohit Singh.
\newblock Online node-weighted steiner tree and related problems.
\newblock In {\em Proceedings of the 2011 IEEE 52Nd Annual Symposium on
  Foundations of Computer Science}, FOCS '11, pages 210--219, Washington, DC,
  USA, 2011. IEEE Computer Society.

\end{thebibliography}

\end{document}